\documentclass[10pt]{article}
\usepackage{amsmath}
\usepackage{amssymb}
\usepackage{color}
\usepackage{float}
\usepackage{graphicx}
\usepackage{hyperref}
\usepackage{multicol}
\usepackage{setspace}
\usepackage{textcomp}
\usepackage{verbatim}
\usepackage{authblk}

{\makeatletter
 \gdef\xxxmark{%
   \expandafter\ifx\csname @mpargs\endcsname\relax % in minipage?
     \expandafter\ifx\csname @captype\endcsname\relax % in figure/caption?
       \marginpar{xxx}% not in a caption or minipage, can use marginpar
     \else
       xxx % notice trailing space
     \fi
   \else
     xxx % notice trailing space
   \fi}
 \gdef\xxx{\@ifnextchar[\xxx@lab\xxx@nolab}
 \long\gdef\xxx@lab[#1]#2{{\bf [\xxxmark #2 ---{\sc #1}]}}
 \long\gdef\xxx@nolab#1{{\bf [\xxxmark #1]}}
 \gdef\turnoffxxx{\long\gdef\xxx@lab[##1]##2{}\long\gdef\xxx@nolab##1{}}%
}

\newtheorem{theorem}{Theorem}
\newtheorem{lemma}{Lemma}

\newcommand{\sq}{\hbox{\rlap{$\sqcap$}$\sqcup$}}
\newcommand{\qed}{\hspace*{\fill}\sq}
\newenvironment{proof}{\noindent\textbf{Proof.}\ }{\qed\par\vskip 4mm\par}

\floatstyle{ruled}
\newfloat{algo}{bhpt}{lop}
\floatname{algo}{Algorithm}

\title{A Super-Fast Distributed Algorithm for Bipartite Metric Facility Location}
\date{}

\author{James Hegeman}
\author{Sriram V.~Pemmaraju
\thanks{This work is supported in part by National Science Foundation grant CCF
0915543. This is a full version of a paper that appeared in DISC 2013. It contains proofs that had
to be removed from the DISC 2013 paper due to space constraints.}} 
\affil{Department of Computer Science\\
The University of Iowa\\
Iowa City, Iowa 52242-1419, USA\\
\texttt{[james-hegeman,sriram-pemmaraju]@uiowa.edu}
}

\begin{document}

\maketitle

\begin{abstract}
The \textit{facility location} problem consists of a set of \textit{facilities}
$\mathcal{F}$, a set of \textit{clients} $\mathcal{C}$, an \textit{opening cost}
$f_i$ associated with each facility $x_i$, and a \textit{connection cost}
$D(x_i,y_j)$ between each facility $x_i$ and client $y_j$. The goal is to find a
subset of facilities to \textit{open}, and to connect each client to an open
facility, so as to minimize the total facility opening costs plus connection
costs. This paper presents the first expected-sub-logarithmic-round distributed
$O(1)$-approximation algorithm in the $\mathcal{CONGEST}$ model for the
\textit{metric} facility location problem on the complete bipartite network with
parts $\mathcal{F}$ and $\mathcal{C}$. Our algorithm has an expected running
time of $O((\log \log n)^3)$ rounds, where $n = |\mathcal{F}| + |\mathcal{C}|$.
This result can be viewed as a continuation of our recent work (ICALP 2012) in
which we presented the first sub-logarithmic-round distributed
$O(1)$-approximation algorithm for metric facility location on a \textit{clique}
network. The bipartite setting presents several new challenges not present in
the problem on a clique network. We present two new techniques to overcome these
challenges. (i) In order to deal with the problem of not being able to choose
appropriate probabilities (due to lack of adequate knowledge), we design an
algorithm that performs a random walk over a probability space and analyze the
progress our algorithm makes as the random walk proceeds. (ii) In order to deal
with a problem of quickly disseminating a collection of messages, possibly
containing many duplicates, over the bipartite network, we design a
probabilistic hashing scheme that delivers all of the messages in
expected-$O(\log \log n)$ rounds.
\end{abstract}

\section{Introduction}
\label{section:Introduction}

This paper continues the recently-initiated exploration
\cite{BHP12,BHP12arxiv,LenzenArxiv2012,LPPP05,PattShamirTeplitsky11} of the design of
sub-logarithmic, or ``super-fast'' distributed algorithms in low-diameter,
bandwidth-constrained settings. To understand the main themes of this
exploration, suppose that we want to design a distributed algorithm for a
problem on a low-diameter network (we have in mind a clique network or a
diameter-$2$ network). In one sense, this is a trivial task since the entire
input could be shipped off to a single node in a single round and that node can simply
solve the problem locally. On the other hand, the problem could be quite
challenging if we were to impose reasonable constraints on bandwidth that
prevent the fast delivery of the entire input to a small number of nodes. A
natural example of this phenomenon is provided by the
\textit{minimum spanning tree} (MST) problem. Consider a clique network in which
each edge $(u,v)$ has an associated weight $w(u,v)$ of which only the nodes $u$
and $v$ are aware. The problem is for the nodes to compute an MST of the
edge-weighted clique such that after the computation, each node knows all MST
edges. It is important to note that the problem is defined by $\Theta(n^2)$
pieces of input and it would take $\Omega\left(\frac{n}{B}\right)$ rounds of
communication for all of this information to reach a single node (where $B$ is
the number of bits that can travel across an edge in each round). Typically,
$B = O(\log n)$, and this approach is clearly too slow given our goal of
completing the computation in a sub-logarithmic number of rounds. Lotker et
al.~\cite{LPPP05} showed that the MST problem on a clique can in fact be solved
in $O(\log \log n)$ rounds in the $\mathcal{CONGEST}$ model of distributed
computation, which is a synchronous, message-passing model in which each node
can send a message of size $O(\log n)$ bits to each neighbor in each round. The
algorithm of Lotker et al.~employs a clever merging procedure that, roughly
speaking, causes the sizes of the MST components to square with each iteration,
leading to an $O(\log \log n)$-round computation time. The overall challenge in
this area is to establish the round complexity of a variety of problems that
make sense in low-diameter settings. The area is largely open with few upper
bounds and no non-trivial lower bounds known. For example, it has been proved
that computing an MST requires
$\Omega\left((\frac{n}{\log n})^{1/4}\right)$ rounds in the
$\mathcal{CONGEST}$ model for diameter-$3$ graphs \cite{LotkerPP06}, but no
lower bounds are known for diameter-$2$ or diameter-$1$ (clique) networks.

The focus of this paper is the \textit{distributed facility location} problem,
which has been considered by a number of researchers
\cite{MoscibrodaWattenhofer05,GLS06,PanditPemmaraju09,PanditPemmaraju10,BHP12}
in low-diameter settings. We first describe the sequential version of the
problem. The input to the facility location problem consists of a set of
\textit{facilities} $\mathcal{F} = \{x_1, x_2, \ldots, x_{n_f}\}$, a set of
\textit{clients} $\mathcal{C} = \{y_1, y_2, \ldots, y_{n_c}\}$, a (nonnegative)
\textit{opening cost} $f_i$ associated with each facility $x_i$, and a
(nonnegative) \textit{connection cost} $D(x_i,y_j)$ between each facility $x_i$
and client $y_j$. The goal is to find a subset $F \subseteq \mathcal{F}$ of
facilities to \textit{open} so as to minimize the total facility opening costs
plus connection costs, i.e.
$FacLoc(F) := \sum_{x_i \in F} f_i + \sum_{y_j \in \mathcal{C}} D(F,y_j)$, where
$D(F,y_j) := \min_{x_i \in F} D(x_i,y_j)$. Facility location is an old and
well-studied problem in operations research
\cite{Balinski66,CNWBook,EHK77,HamburgerKuehn63,StollSteimer63} that arises in
contexts such as locating hospitals in a city or locating distribution centers
in a region. The \textit{metric facility location} problem is an important
special case of facility location in which the connection costs satisfy the
following ``triangle inequality:'' for any $x_i, x_{i'} \in \mathcal{F}$ and
$y_j, y_{j'} \in \mathcal{C}$,
$D(x_i,y_j) + D(y_j,x_{i'}) + D(x_{i'},y_{j'}) \geq D(x_i,y_{j'})$. The facility
location problem, even in its metric version, is NP-complete and finding
approximation algorithms for the problem has been a fertile area of research.
There are several constant-factor approximation algorithms for metric facility
location (see \cite{Li11} for a recent example). This approximation factor is
known to be near-optimal \cite{GuhaKhuller98}.

More recently, the facility location problem has also been used as an
abstraction for the problem of locating resources in wireless networks
\cite{FrankBook,PanditPemmarajuICDCN09}. Motivated by this application, several
researchers have considered the facility location problem in a distributed
setting. In \cite{MoscibrodaWattenhofer05,PanditPemmaraju09,PanditPemmaraju10},
as well as in the present work, the underlying communication network is a
complete bipartite graph $G = \mathcal{F} + \mathcal{C}$, with $\mathcal{F}$ and
$\mathcal{C}$ forming the bipartition. At the beginning of the algorithm, each
node, whether a facility or client, has knowledge of the connection costs
(``distances'') between itself and all nodes in the other part. In addition, the
facilities know their opening costs. The problem is to design a distributed
algorithm that runs on $G$ in the $\mathcal{CONGEST}$ model and produces a
subset $F \subseteq \mathcal{F}$ of facilities to \textit{open}. To simplify
exposition we assume that every cost in the problem input can be represented in
$O(\log n)$ bits, thus allowing each cost to be transmitted in a single message.
Each chosen facility will then open and provide services to any and all clients
that wish to connect to it (each client must be served by some facility). The
objective is to guarantee that $FacLoc(F) \leq \alpha \cdot OPT$, where $OPT$ is
the cost of an optimal solution to the given instance of facility location and
$\alpha$ is a constant. We call this the \textsc{BipartiteFacLoc} problem. In
this paper we present the first sub-logarithmic-round algorithm for the
\textsc{BipartiteFacLoc} problem; specifically, our algorithm runs in
$O((\log \log n_f)^2 \cdot \log \log \min \{n_f, n_c\})$ rounds in expectation,
where $n_f = |\mathcal{F}|$ and $n_c = |\mathcal{C}|$. All previous distributed
approximation algorithms for \textsc{BipartiteFacLoc} require a logarithmic
number of rounds to achieve near-optimal approximation factors.

\subsection{Overview of Technical Contributions}

In a recent paper (ICALP 2012, \cite{BHP12}; full version available as \cite{BHP12arxiv}), we presented an
expected-$O(\log \log n)$-round algorithm in the $\mathcal{CONGEST}$ model for
\textsc{CliqueFacLoc}, the ``clique version'' of \textsc{BipartiteFacLoc}. The
underlying communication network for this version of the problem is a clique
with each edge $(u,v)$ having an associated (connection) cost $c(u,v)$ of which
only nodes $u$ and $v$ are aware (initially). Each node $u$ also has an opening
cost $f_u$, and may choose to open as a facility; nodes that do not open must
connect to an open facility. The cost of the solution is defined as before -- as
the sum of the facility opening costs and the costs of established connections.
Under the assumption that the connection costs form a metric, our algorithm for
\textsc{CliqueFacLoc} yields an $O(1)$-approximation. We had hoped that a
``super-fast'' algorithm for \textsc{BipartiteFacLoc} would be obtained in a
straightforward manner by extending our \textsc{CliqueFacLoc} algorithm.
However, it turns out that moving from a clique communication network to a
complete bipartite communication network raises several new and significant
challenges related to information dissemination and a lack of adequate
knowledge. Below we outline these challenges and our solutions to them.

\textbf{Overview of solution to \textsc{CliqueFacLoc}.} To solve
\textsc{CliqueFacLoc} on an edge-weighted clique $G$ \cite{BHP12,BHP12arxiv} we reduce it
to the problem of computing a $2$-ruling set in an appropriately-defined
spanning subgraph of $G$. A \textit{$\beta$-ruling set} of a graph is an
independent set $S$ such that every node in the graph is at most $\beta$ hops
away from some node in $S$; a \textit{maximal independent set} (MIS) is simply a
$1$-ruling set. The spanning subgraph $H$ on which we compute a $2$-ruling set
is induced by clique edges whose costs are no greater than a pre-computed
quantity which depends on the two endpoints of the edge in question.

We solve the $2$-ruling set problem on the spanning subgraph $H$ via a
combination of deterministic and randomized sparsification. Briefly, each node
selects itself with a uniform probability $p$ chosen such that the subgraph $H'$
of $H$ induced by the selected nodes has $\Theta(n)$ edges in expectation. The
probability $p$ is a function of $n$ and the number of edges in $H$. We next
deliver all of $H'$ to every node.
It can be shown that a
graph with $O(n)$ edges can be completely delivered to every node in $O(1)$
rounds on a clique and since $H'$ has $O(n)$ edges in expectation, the delivery
of $H'$ takes expected-$O(1)$ rounds. 
Once $H'$ has been disseminated in this manner, each node uses the same
(deterministic) rule to locally compute an MIS of $H'$. 
Following the computation of an MIS of
$H'$, nodes in the MIS and nodes in their 2-neighborhood are all deleted from
$H$ and $H$ shrinks in size. Since $H$ is now smaller, a larger probability $p$ can be used
for the next iteration. This increasing sequence of values for $p$ results in a
doubly-exponential rate of progress, which leads to an
expected-$O(\log \log n)$-round algorithm for computing a $2$-ruling set of $H$.
See \cite{BHP12} for more details.

\textbf{Challenges for \textsc{BipartiteFacLoc}.} The same algorithmic framework
can be applied to \textsc{BipartiteFacLoc}; however, challenges arise in trying
to implement the ruling-set computation on a bipartite communication network. As
in \textsc{CliqueFacLoc} \cite{BHP12}, we define a particular graph $H$ on the
set of facilities with edges connecting pairs of facilities whose connection
cost is bounded above. Note that there is no explicit notion of connection cost
between facilities, but we use a natural extension of the facility-client
connection costs $D(\cdot,\cdot)$ and define for each
$x_i, x_j \in \mathcal{F}$,
$D(x_i,x_j) := \min_{y \in \mathcal{C}} D(x_i,y) + D(x_j,y)$. The main
algorithmic step now is to compute a $2$-ruling set on the graph $H$. However,
difficulties arise because $H$ is not a subgraph of the communication network $G$, as it was in the
\textsc{CliqueFacLoc} setting. In fact, initially a facility $x_i$ does not even
know to which other facilities it is adjacent to in $H$. This adjacency
knowledge is collectively available only to the clients. A client $y$
\textit{witnesses} edge $\{x_i,x_j\}$ in $H$ if $D(x_i, y) + D(x_j, y)$ is bounded
above by a pre-computed quantity associated with the facility-pair $x_i, x_j$.
However, (initially) an individual client $y$ cannot certify the
\textit{non-existence} of any potential edge between two facilities in $H$; as,
unbeknownst to $y$, some other client may be a witness to that edge.
Furthermore, the same edge $\{x_i, x_j\}$ could have many client-witnesses. This
``affirmative-only'' adjacency knowledge and the duplication of this knowledge
turn out to be key obstacles to overcome. For example, in this setting, it seems
difficult to even figure out how many edges $H$ has.

Thus, an example of a problem we need to solve is this: without knowing the
number of edges in $H$, how do we correctly pick a probability $p$ that will
induce a random subgraph $H'$ with $\Theta(n)$ edges? Duplication of knowledge
of $H$ leads to another problem as well. Suppose we did manage to pick a
``correct'' value of $p$ and have induced a subgraph $H'$ having $\Theta(n)$
edges. In the solution to \textsc{CliqueFacLoc}, we were able to deliver all of
$H'$ to a single node (in fact, to every node). In the bipartite setting, how do we deliver $H'$ to a
single node given that even though it has $O(n)$ edges, information duplication
can cause the sum of the number of adjacencies witnessed by the clients to be as high as
$\Omega(n^2)$?

We introduce new techniques to solve each of these problems. These techniques
are sketched below.

\begin{itemize}
\item\textbf{Message dissemination with duplicates.} We model the problem of
delivering all of $H'$ to a single node as the following message-dissemination
problem on a complete bipartite graph.
\vspace{2mm}

\begin{quote}
\textbf{Message Dissemination with Duplicates (MDD).}\\
Given a bipartite graph $G = \mathcal{F} + \mathcal{C}$, with
$n_f := |\mathcal{F}|$ and $n_c := |\mathcal{C}|$, suppose that there are $n_f$
messages that we wish to be known to all client nodes in $\mathcal{C}$.
Initially, each client possesses some subset of the $n_f$ messages, with each
message being possessed by at least one client. Suppose, though, that no client
$y_j$ has any information about which of its messages are also held by any other
client. Disseminate all $n_f$ messages to each client in the network in
expected-sub-logarithmic time.
\end{quote}
\vspace{2mm}

We solve this problem by presenting an algorithm that utilizes probabilistic
hashing to iteratively reduce the number of duplicate copies of each message.
Note that if no message exists in duplicate, then the total number of messages
held is only $n_f$, and each can be sent to a distinct facility which can then
broadcast it to every client. The challenge, then, lies in coordinating
bandwidth usage so as to avoid ``bottlenecks'' that could be caused by message
duplication.
Our algorithm for MDD runs in $O(\log\log \min\{n_f, n_c\})$ rounds in expectation.
\vspace{2mm}

\item\textbf{Random walk over a probability space.} Given the difficulty of
quickly acquiring even basic information about $H$ (e.g., how many edges does it
have?), we have no way of setting the value of $p$ correctly. So we design an
algorithm that performs a random walk over a space of $O(\log \log n_f)$
probabilities. The algorithm picks a probability $p$, uses this to induce a
random subgraph $H'$ of $H$, and attempts to disseminate $H'$ to all clients
within $O(\log \log \min \{n_f, n_c\})$ rounds. If this dissemination succeeds,
$p$ is modified in one way (increased appropriately), otherwise $p$ is modified differently (decreased appropriately). 
This technique can be modeled as a random walk on a probability space consisting of
$O(\log\log n_f)$ elements, where the elements are distinct values that $p$ can take.
We show that after a random walk of length at most $O(\log \log n_f)$, sufficiently 
many edges of $H$ are removed, leading to $O(\log\log n_f)$ levels of progress.
Thus we have a total of $O((\log \log n_f)^2)$ steps and since in each step 
an instance of MDD is solved for disseminating adjacencies, we obtain 
an expected-$O((\log \log n_f)^2 \cdot \log \log \min \{n_f, n_c\})$-round
algorithm for computing a $2$-ruling set of $H$.

\end{itemize}

To summarize, our paper makes three main technical contributions. (i) We show
(in Section \ref{section:Reduction}) that the framework developed in
\cite{BHP12} to solve \textsc{CliqueFacLoc} can be used, with appropriate
modifications, to solve \textsc{BipartiteFacLoc}. Via this algorithmic
framework, we reduce \textsc{BipartiteFacLoc} to the problem of computing a
$2$-ruling set of a graph induced by facilities in a certain way. (ii) In order
to compute a 2-ruling set of a graph, we need to disseminate graph adjacencies
whose knowledge is distributed among the clients with possible duplication. We
model this as a message dissemination problem and show (in Section
\ref{section:Dissemination}), using a probabilistic hashing scheme, how to
efficiently solve this problem on a complete bipartite graph. (iii) Finally, we
present (in Section \ref{section:2RulingSet}) an algorithm that performs a
random walk over a probability space to efficiently compute a 2-ruling set of a
graph, without even basic information about the graph. This algorithm repeatedly
utilizes the procedure for solving the message-dissemination problem mentioned
above.
\vspace{2mm}

\section{Reduction to the Ruling Set Problem}
\label{section:Reduction}

In this section we reduce \textsc{BipartiteFacLoc} to the ruling set problem on
a certain graph induced by facilities. The reduction is achieved via the
distributed facility location algorithm called \textsc{LocateFacilities} and
shown as Algorithm \ref{alg:bipartite_facloc}. This algorithm is complete except
that it calls a subroutine, \textsc{RulingSet}$(H, s)$ (in Step 4), to compute
an $s$-ruling set of a certain graph $H$ induced by facilities. In this section
we first describe Algorithm \ref{alg:bipartite_facloc} and then present its
analysis. It is easily observed that all the steps in Algorithm
\ref{alg:bipartite_facloc}, except the one that calls \textsc{RulingSet}$(H, s)$
take a total of $O(1)$ communication rounds. Thus the running time of
\textsc{RulingSet}$(H, s)$ essentially determines the running time of Algorithm
\ref{alg:bipartite_facloc}. Furthermore, we show that if $F^*$ is the subset of
facilities opened by Algorithm \ref{alg:bipartite_facloc}, then
$FacLoc(F^*) = O(s) \cdot OPT$. In the remaining sections of the paper we show
how to implement \textsc{RulingSet}$(H, 2)$ in expected
$O((\log\log n_f)^2 \cdot \log\log \min\{n_f, n_c\})$ rounds. This yields an
expected $O((\log\log n_f)^2 \cdot \log\log \min\{n_f, n_c\})$-round,
$O(1)$-approximation algorithm for \textsc{BipartiteFacLoc}.

\subsection{Algorithm}
\label{subsection:algorithm}

Given $\mathcal{F}$, $\mathcal{C}$, $D(\cdot,\cdot)$, and $\{f_i\}$, define the
\textit{characteristic radius} $r_i$ of facility $x_i$ to be the nonnegative
real number satisfying $\sum_{y \in B(x_i,r_i)} (r_i - D(x_i,y)) = f_i$, where
$B(x,r)$ (the \textit{ball} of radius $r$) denotes the set of clients $y$ such
that $D(x,y) \leq r$. This notion of a characteristic radius was first
introduced by Mettu and Plaxton \cite{MettuPlaxton03}, who use it to drive their
sequential, greedy algorithm. We extend the client-facility distance function
$D(\cdot,\cdot)$ to facility-facility distances; let
$D: \mathcal{F} \times \mathcal{F} \rightarrow \mathbb{R}^+ \cup \{0\}$ be
defined by
$D(x_i,x_j) = \min_{y_k \in \mathcal{C}} \{D(x_i,y_k) + D(x_j,y_k)\}$. 
With these definitions in place we are ready to describe Algorithm
\ref{alg:bipartite_facloc}. The algorithm consists of three stages, which we now
describe.

\begin{algo}
\textbf{Input:} A complete bipartite graph $G$ with partition $(\mathcal{F}, \mathcal{C})$; 
(bipartite) metric $D(\cdot,\cdot)$; opening costs $\{f_i\}_{i=1}^{n_f}$; a
sparsity parameter $s \in \mathbb{Z}^+$\\
\textbf{Assumption:} Each facility knows its own opening cost and its distances
to all clients; each client knows its distances to all facilities\\
\textbf{Output:} A subset of facilities (a \textit{configuration}) to be
declared open.
{\small
\begin{tabbing}
......\=a..\=b..\=c..\=d..\=e..\=f..\=g..\=h..\=i..\=j..\=k..\=l\kill
1.\>Each facility $x_i$ computes and broadcasts its radius $r_i$ to all clients;
$r_0 := \min_i r_i$.\\
2.\>Each client computes a partition of the facilities into classes $\{V_k\}$
such that $3^k \cdot r_0 \leq r_i < 3^{k+1} \cdot r_0$ for $x_i \in V_k$.\\
3.\>For $k = 0, 1, \ldots$, define a graph $H_k$ with vertex set $V_k$ and edge
set:\\
\>\>$\{\{x_i, x_{i'}\} \mid x_i, x_{i'} \in V_k \mbox{ and }
D(x_i,x_{i'}) \leq r_i + r_{i'}$\}\\
\>\>(Observe from the definition of facility distance that such edges may be
known to as few as one client,\\
\>\>or as many as all of them.)\\
4.\>All nodes in the network use procedure \textsc{RulingSet}($\cup_k H_k, s$)
to compute a 2-ruling set $T$ of $\cup_k H_k$.\\
\>\>$T$ is known to every client. We use $T_k$ to denote $T \cap V_k$.\\
5.\>Each client $y_j$ sends an \texttt{open} message to each facility $x_i$, if
and only if both of the following conditions hold:\\
\>\>(i) $x_i$ is a member of the set $T_k \subseteq H_k$, for some $k$.\\
\>\>(ii) $y_j$ is not a witness to the existence of a facility $x_{i'}$
belonging to a class $H_{k'}$, with $k' < k$,\\
\>\>\>such that $D(x_i,x_{i'}) \leq 2 r_i$.\\
6.\>Each facility $x_i$ opens, and broadcasts its status as such, if and only
if $x_i$ received an \texttt{open} message from\\
\>\>every client.\\
7.\>Each client connects to the nearest open facility.
\end{tabbing}}
\caption{\textsc{LocateFacilities}}
\label{alg:bipartite_facloc}
\end{algo}

\paragraph{Stage 1. (Steps 1-2)} Each facility knows its own opening cost and the distances
to all clients. So in Step 1 facility $x_i$ computes $r_i$ and broadcasts that
value to all clients. Once this broadcast is complete, each client knows all of
the $r_i$ values. This enables every client to compute the same partition of the
facilities into classes as follows (Step 2). Define the special value
$r_0 := \min_{1 \leq i \leq n_f} \{r_i\}$. Define the class $V_k$, for
$k = 0, 1, \ldots$, to be the set of facilities $x_i$ such that
$3^k \cdot r_0 \leq r_i < 3^{k+1} \cdot r_0$. Every client computes the class
into which each facility in the network falls.

\paragraph{Stage 2. (Steps 3-4)} Now that the facilities are divided into classes having
comparable $r_i$'s, and every client knows which facility is in each class, we
focus our attention on class $V_k$. Suppose $x_i, x_{i'} \in V_k$. Then we
define $x_i$ and $x_{i'}$ to be \textit{adjacent} in class $V_k$ if
$D(x_i,x_{i'}) \leq r_i + r_{i'}$ (Step 3). These adjacencies define the graph
$H_k$ with vertex set $V_k$. Note that two facilities $x_i$, $x_{i'}$ in class
$V_k$ are adjacent if and only if there is at least one client \textit{witness}
for this adjacency. Next, the network computes an $s$-ruling set $T$ of $\cup_k H_k$
with procedure \textsc{RulingSet}() (Step 4).
We describe a super-fast implementation of \textsc{RulingSet}() in Section
\ref{section:2RulingSet}. After a ruling set $T$ has been constructed, 
every client knows all the members of $T$.
Since the $H_k$'s are disjoint, $T_k := T \cap V_k$ is a 2-ruling set of $H_k$ for each $k$.

\paragraph{Stage 3. (Steps 5-7)} Finally, a client $y_j$ sends an \texttt{open} message to
facility $x_i$ in class $V_k$ if (i) $x_i \in T_k$, and (ii) there is no
facility $x_{i'}$ of class $V_{k'}$ such that
$D(x_i,y_j) + D(x_{i'},y_j) \leq 2 r_i$, and for which $k' < k$ (Step 5). A
facility opens if it receives \texttt{open} messages from all clients (Step 6).
Lastly, open facilities declare themselves as such in a broadcast, and every
client connects to the nearest open facility (Step 7).

Algorithm \textsc{LocateFacilities} is complete except for the call to the \textsc{RulingSet} procedure.
The remaining sections of the paper describe the implementation and analysis of \textsc{RulingSet}.

\subsection{Analysis}

The approximation-factor analysis of Algorithm \ref{alg:bipartite_facloc} is
similar to the analysis of our algorithm for \textsc{CliqueFacLoc} \cite{BHP12,BHP12arxiv}.
First, we show a lower bound on the cost of \textit{any} solution to
\textsc{BipartiteFacLoc}. In order to do so, we define $\overline{r}_j$ (for
$y_j \in \mathcal{C}$) as
$\overline{r}_j = \min_{1 \leq i \leq n_f} \{r_i + D(x_i,y_j)\}$. This concept
was introduced and motivated in \cite{BHP12,BHP12arxiv}.
Specifically, we show that the cost of any solution to the facility location problem
is bounded before by $\frac{1}{6} \cdot \sum_{j=1}^{n_c} \overline{r}_j$.
Subsequently, we show that the solution computed by Algorithm \textsc{LocateFacilities}
has cost that is $O(s)$ times $\sum_{j=1}^{n_c} \overline{r}_j$.
Thus, guaranteeing $s = O(1)$, yields an $O(1)$-approximation.

%The $\overline{r}_j$'s play a key role in the analysis, beginning in Lemma
%\ref{lemma:cost_bound}.

\subsubsection{Approximation Analysis - Lower Bound}

We start the lower bound proof by extending the sequential metric facility
location algorithm (and analysis) of Mettu and Plaxton \cite{MettuPlaxton03} on a clique
network to the bipartite setting. This part of the analysis closely follows
\cite{MettuPlaxton03} and we include it mainly for completeness.
We start by presenting the bipartite version of the Mettu-Plaxton algorithm.
The algorithm is greedy in that it considers facilities in non-decreasing order
of their $r$-values and opens a facility only if there is no already-open
facility within 2 times the $r$-value of the facility being considered.
Below we use the $D(x, F)$, where $x \in \mathcal{F}$ and $F \subseteq \mathcal{F}$,
to denote $\min_{x' \in F} D(x, x')$.

\begin{algo}
\textbf{Input:} $\mathcal{F}$, $\mathcal{C}$, $D(\cdot,\cdot)$, $\{f_i\}$\\
\textbf{Output:} A subset of facilities to open
\small{
\begin{tabbing}
......\=a..\=b..\=c..\=d..\=e..\=f..\=g..\=h..\=i..\=j..\=k..\=l\kill
1.\>Let $F_0 = \emptyset$.\\
2.\>For each facility $x_i$, compute the characteristic radius $r_i$.\\
3.\>Let $\varphi$ be a permutation of $\{1, \ldots, n_f\}$ such that for $1 \leq i < i' \leq n_f$,
$r_{\varphi(i)} \leq r_{\varphi(i')}$.\\
4.\>For $i = 1$ to $n_f$, if $D(x_{\varphi(i)},F_{i-1}) > 2 r_{\varphi(i)}$, then set
$F_i = F_{i-1} \cup \{x_{\varphi(i)}\}$; else set $F_i = F_{i-1}$.\\
5.\>Return $F_{MP} = F_{n_f}$.
\end{tabbing}}
\caption{Bipartite Mettu-Plaxton Algorithm}
\label{alg:BipartiteMP}
\end{algo}

\noindent The running time of Algorithm \ref{alg:BipartiteMP} is not important
to us, but the approximation factor is. Let $F_{MP}$ denote the subset of
facilities opened by algorithm \ref{alg:BipartiteMP}. 
The following series of lemmas (Lemmas \ref{lemma:select_nbr} to \ref{lemma:compare_charge}) lead to Theorem \ref{theorem:MPcost},
which shows that $FacLoc(F_{MP})$ is within 3 times the optimal facility opening cost.

%Lemma 1
\begin{lemma}
For any facility $x_i$, there exists a facility $x_j \in F_{MP}$ such that
$\varphi^{-1}(j) \leq \varphi^{-1}(i)$ (i.e. $r_j \leq r_i$) and
$D(x_i,x_j) \leq 2 r_i$. (Note that $x_j$ may be $x_i$ itself.)
\label{lemma:select_nbr}
\end{lemma}
\begin{proof}
Suppose not. Then $D(x_i,F_{\varphi^{-1}(i)-1}) > 2 r_i$, so $x_i$ should have
been added to $F_{MP}$, which is a contradiction.
\end{proof}

%Lemma 2
\begin{lemma}
Let $x_i, x_j \in F_{MP}$. Then $D(x_i,x_j) > 2 \cdot \max \{r_i, r_j\}$.
\label{lemma:MP_sparse}
\end{lemma}
\begin{proof}
Without loss of generality, assume that $\varphi^{-1}(i) < \varphi^{-1}(j)$.
Then $r_i \leq r_j$, and since $x_j$ was added to $F_{MP}$, it must have been
the case that $D(x_j,F_{\varphi^{-1}(j)-1})$ was greater than $2 r_j$. As
$x_i \in F_{\varphi^{-1}(j)-1}$, we conclude that
$D(x_i,x_j) > 2 r_j = 2 \cdot \max\{r_i, r_j\}$.
\end{proof}

An important contribution of \cite{MettuPlaxton03} was a standard way of
``charging'' the cost of a facility location solution to clients.
For a client $y_j \in \mathcal{C}$ and a facility subset $F$, the
\textit{charge} of $y_j$ with respect to $F$ is defined as
\[charge(y_j,F) = D(F,y_j) + \sum\limits_{x_i \in F} \max \{0, r_i - D(x_i, y_j)\}\]
Here $D(F, y)$, for $F \subseteq \mathcal{F}$ and $y \in \mathcal{C}$, is used as shorthand for
$\min_{x \in \mathcal{F}} D(x, y)$.
In the following lemma, we use simple algebraic manipulation to show that for any facility subset
$F$, the cost of $F$ is correctly distributed to the ``charge'' associated with each client,
as per the definition  $charge(y_j,F)$.

%Lemma 3
\begin{lemma}
For any subset $F$, $\sum_{y_j \in \mathcal{C}} charge(y_j,F) = FacLoc(F)$.
\label{lemma:sum_charge}
\end{lemma}
\begin{proof}
Observe that
\begin{align*}
\sum\limits_{y_j \in \mathcal{C}} charge(y_j,F) &= %
\sum\limits_{y_j \in \mathcal{C}} D(F,y_j) \; + \; %
\sum\limits_{x_i \in F} \sum\limits_{y_j \in \mathcal{C}} %
\max \{0, r_i - D(x_i,y_j)\}\\
&= \sum\limits_{y_j \in \mathcal{C}} D(F,y_j) \; + \; %
\sum\limits_{x_i \in F} \sum\limits_{y_j \in B(x_i,r_i)} (r_i - D(x_i,y_j))\\
&= \sum\limits_{y_j \in \mathcal{C}} %
D(F,y_j) \; + \; \sum\limits_{x_i \in F} f_i\\[2mm]
&= FacLoc(F)\\
\end{align*}
\end{proof}

%Lemma 4
\begin{lemma}
Let $y_j$ be a client, let $F$ be a subset of facilities, and let
$x_i \in F$. If $D(x_i,y_j) = D(F,y_j)$, then
$charge(y_j,F) \geq \max\{r_i, D(x_i,y_j)\}$.
\label{lemma:charge_any}
\end{lemma}
\begin{proof}
If $D(x_i,y_j) > r_i$, then $charge(y_j,F) \geq D(F,y_j) = D(x_i,y_j) > r_i$. If
$D(x_i,y_j) \leq r_i$, then $charge(y_j,F) \geq D(F,y_j) + (r_i - D(x_i,y_j)) =
D(x_i,y_j) + (r_i - D(x_i,y_j)) = r_i \geq D(x_i,y_j)$.
\end{proof}

%Lemma 5
\begin{lemma}
Let $y_j$ be a client and let $x_i \in F_{MP}$. If $y_j \in B(x_i,r_i)$, then
$charge(y_j,F_{MP}) = r_i$.
\label{lemma:charge_near}
\end{lemma}
\begin{proof}
By Lemma \ref{lemma:MP_sparse}, there can be no other facility
$x_{i'} \in F_{MP}$, $i' \neq i$, such that
$D(x_{i'},y_j) \leq \max \{r_i, r_{i'}\}$, for then $D(x_i,x_{i'})$ would be at
most $2 \cdot \max \{r_i, r_{i'}\}$. Therefore
$charge(y_j,F_{MP}) = D(x_i,y_j) + (r_i - D(x_i,y_j)) = r_i$.
\end{proof}

%Lemma 6
\begin{lemma}
Let $y_j$ be a client and let $x_i \in F_{MP}$. If $y_j \notin B(x_i,r_i)$, then
$charge(y_j,F_{MP}) \leq D(x_i,y_j)$.
\label{lemma:charge_far}
\end{lemma}
\begin{proof}
If there is no $x_{i'} \in F_{MP}$ such that $y_j \in B(x_{i'},r_{i'})$, then
$charge(y_j,F_{MP}) = D(F_{MP},y_j) \leq D(x_i,y_j)$. If there is such an
$x_{i'}$, then by Lemma \ref{lemma:MP_sparse},
$D(x_i,x_{i'}) > 2 \cdot \max \{r_i, r_{i'}\}$. By Lemma
\ref{lemma:charge_near}, then,
\begin{align*}
charge(y_j,F_{MP}) &= r_{i'}\\[1mm]
&\leq D(x_i,x_{i'}) - r_{i'}\\[1mm]
&\leq D(x_i,x_{i'}) - D(x_{i'},y_j)\\[1mm]
&\leq (D(x_i,x_{i'}) - D(x_{i'},y_j) - D(x_i,y_j)) + D(x_i,y_j)\\[1mm]
&\leq D(x_i,y_j)\\
\end{align*}
\end{proof}

%Lemma 7
\begin{lemma}
For any client $y_j$ and subset $F$,
$charge(y_j,F_{MP}) \leq 3 \cdot charge(y_j,F)$.
\label{lemma:compare_charge}
\end{lemma}
\begin{proof}
Let $x_i \in F$ be such that $D(x_i,y_j) = D(F,y_j)$. By Lemma
\ref{lemma:select_nbr}, there is a facility $x_{i'} \in F_{MP}$ such that
$\varphi^{-1}(i') \leq \varphi^{-1}(i)$ ($r_{i'} \leq r_i$) and
$D(x_i,x_{i'}) \leq 2 r_i$.

If $y_j \in B(x_{i'},r_{i'})$, then by Lemma \ref{lemma:charge_near} we have
$charge(y_j,F_{MP}) = r_{i'} \leq r_i$; thus, by Lemma \ref{lemma:charge_any},\\
$charge(y_j,F_{MP}) \leq charge(y_j,F)$.

If $y_j \notin B(x_{i'},r_{i'})$, then by Lemma \ref{lemma:charge_far} we have
$charge(y_j,F_{MP}) \leq D(x_{i'},y_j) \leq D(x_{i'},x_i) + D(x_i,y_j) \leq
2 r_i + D(x_i,y_j)$. Now, by Lemma \ref{lemma:charge_any}, we see that
$2 r_i + D(x_i,y_j) \leq 3 \cdot \max \{r_i, D(x_i,y_j)\} \leq 3 \cdot charge(y_j,F)$.
\end{proof}

\noindent
The following theorem follows from Lemma \ref{lemma:sum_charge} and Lemma \ref{lemma:compare_charge}.

%Theorem 1
\begin{theorem}
For any subset $F$ of facilities, $FacLoc(F_{MP}) \leq 3 \cdot FacLoc(F)$.
\label{theorem:MPcost}
\end{theorem}

\noindent
Now, as mentioned previously, we define $\overline{r}_j$ (for
$y_j \in \mathcal{C}$) as
$\overline{r}_j = \min_{1 \leq i \leq n_f} \{r_i + D(x_i,y_j)\}$.

%Lemma 8
\begin{lemma}
$FacLoc(F) \geq (\sum_{j=1}^{n_c} \overline{r}_j) / 6$ for any subset
$F \subseteq \mathcal{F}$.
\label{lemma:cost_bound}
\end{lemma}
\begin{proof}
Recall that $F_{MP}$ has the property that no two facilities
$x_i, x_j \in F_{MP}$ can have $D(x_i,x_j) \leq r_i + r_j$. Therefore, for a client $y_j$, if
$x_{\delta(j)}$ denotes a closest open facility (i.e. an open facility
satisfying $D(x_{\delta(j),y_j}) = D(F_{MP},y_j)$), then
\begin{align*}
FacLoc(F_{MP}) &= \sum\limits_{j=1}^{n_c} charge(y_j,F_{MP})\\
&\geq \sum\limits_{y_j \in \mathcal{C}} \left[D(x_{\delta(j)},y_j) + %
\max \{0, r_{\delta(j)} - D(x_{\delta(j)},y_j)\}\right]\\
&= \sum\limits_{y_j \in \mathcal{C}} \max \{r_{\delta(j)}, D(x_{\delta(j)},y_j)\}\\
\end{align*}
Note that the inequality in the above calculation follows from throwing away
some terms of the sum in the definition of $charge(y_j,F_{MP})$.

By the definition of $\overline{r}_j$,
$\overline{r}_j \leq r_{\delta(j)} + D(x_{\delta(j)},y_j) \leq
2 \cdot \max \{r_{\delta(j)}, D(x_{\delta(j)},y_j)\}$. It follows that
$$FacLoc(F_{MP}) \geq \sum\limits_{y_j \in \mathcal{C}} \frac{\overline{r}_j}{2}
= \frac{1}{2} \cdot \sum\limits_{j=1}^{n_c} \overline{r}_j.$$
Therefore
$FacLoc(F) \geq FacLoc(F_{MP}) / 3 \geq (\sum_{j=1}^{n_c} \overline{r}_j) / 6$,
for any $F \subseteq \mathcal{F}$.
\end{proof}

\subsubsection{Approximation Analysis - Upper Bound}

Let $F^*$ be the set of facilities opened by Algorithm \ref{alg:bipartite_facloc}.
We analyze $FacLoc(F^*)$ by bounding $charge(y_j,F^*)$ for each client
$y_j$. Recall that $FacLoc(F^*) = \sum_{j=1}^{n_c} charge(y_j,F^*)$. Since
$charge(y_j,F^*)$ is the sum of two terms, $D(F^*,y_j)$ and
$\sum_{x_i \in F^*} \max \{0, r_i - D(x_i,y_j)\}$, bounding each separately by a
$O(s)$-multiple of $\overline{r}_j$, yields the result.

The $s$-ruling set $T_k \subseteq V_k$ has the property that for any node
$x_i \in V_k$, $D(x_i,T_k) \leq 2 \cdot 3^{k+1} r_0 \cdot s$, where $s$ is the
sparsity parameter used to procedure \textsc{RulingSet}(). Also, for no two
members of $T_k$ is the distance between them less than $2 \cdot 3^k r_0$. Note
that here we are using distances from the extension of $D$ to
$\mathcal{F} \times \mathcal{F}$.

Now, in our cost analysis, we consider a facility $x_i \in V_k$. To bound
$D(x_i,F^*)$, observe that either $x_i \in T_k$, or else there exists a facility
$x_{i'} \in T_k$ such that
$D(x_i,x_{i'}) \leq 2 \cdot 3^{k+1} r_0 \cdot s \leq 6 r_i \cdot s$. Also, if a
facility $x_i \in T_k$ does not open, then there exists another node $x_{i'}$ in
a class $V_{k'}$, with $k' < k$, such that $D(x_i,x_{i'}) \leq 2 r_j$.

We are now ready to bound the components of $charge(y_j,F^*)$.

%Lemma 9
\begin{lemma}
$D(F^*,y_j) \leq (15 s + 15) \cdot \overline{r}_j$
\label{lemma:open_fac_dist}
\end{lemma}
\begin{proof}
First, consider any facility $x_i$. Suppose that the class containing $x_i$ is
$V_k$. Observe that the result of procedure \textsc{RulingSet} is that $x_i$ is
within distance $6 r_i \cdot s$ of a facility $x_{i'} \in T_k$ (which may be
$x_i$ itself). Now, in Algorithm \ref{alg:bipartite_facloc}, $x_{i'}$ either
opens, or there exists a facility $x_{i''}$ of a lower class such that
$D(x_{i'},x_{i''}) \leq 2 r_{i'} \leq 6 r_i$. We therefore conclude that within
a distance $(6 s + 6) \cdot r_i$ of $x_i$, there exists either an open facility
or a facility of a class of index less than $k$.

Now, let $x_{j'}$ be a minimizer for $r_x + D(x,y_j)$ so that
$\overline{r}_j = r_{j'} + D(x_{j'},y_j)$, and suppose $x_{j'} \in V_{k'}$. By
the preceding analysis, there exists within a distance $(6 s + 6) \cdot r_{j'}$
of $x_{j'}$ either an open facility or a facility $x_{j''}$ of a class of index
$k'' < k'$. If it is the latter, then within a distance
$(6 s + 6) \cdot r_{j''}$ of $x_{j''}$ there exists either an open facility or a
facility $x_{j'''}$ of a class of index $k''' \leq k' - 2$.

Repeating this argument up to $k' + 1$ times reveals that there must exist an
open facility within a distance
$(6 s + 6) \cdot (r_{j'} + r_{j''} + r_{j'''} + \ldots)$ of $x_{j'}$. (Note that
any facility $x_i$ in class $V_0$ has an open facility within distance
$6 s \cdot r_i$ because every member of $T_0$ opens.) We can simplify this
distance bound by noting that $r_{j'} > r_{j''}$, $r_{j'} > 3 r_{j'''}$,
$r_{j'} > 9 r_{j''''}$, etc., and so $D(F^*,x_{j'}) \leq
(6 s + 6) \cdot (r_{j'} + r_{j'} + \frac{1}{3} r_{j'} + \frac{1}{9} r_{j'} + \ldots) =
(6 s + 6) \cdot \frac{5}{2} r_{j'} = (15 s + 15) r_{j'}$.

Thus we have
\begin{align*}
D(F^*,y_j) &\leq D(F^*,x_{j'}) + D(x_{j'},y_j)\\[1mm]
&\leq (15 s + 15) \cdot r_{j'} + D(x_{j'},y_j)\\[1mm]
&\leq (15 s + 15) \cdot (r_{j'} + D(x_{j'},y_j)\\[1mm]
&= (15 s + 15) \cdot \overline{r}_j\\
\end{align*}
which completes the proof.
\end{proof}

%Lemma 10
\begin{lemma}
$\sum_{x_i \in F^*} \max \{0, r_i - D(x_i,y_j\} \leq 3 \cdot \overline{r}_j$
\label{lemma:open_fac_contrib}
\end{lemma}
\begin{proof}
We begin by observing that we cannot simultaneously have $D(x_i,y_j) \leq r_i$
and $D(x_{i'},y_j) \leq r_{i'}$ for $x_i, x_{i'} \in F^*$ and $i \neq {i'}$.
Indeed, if this were the case, then $D(x_i,x_{i'}) \leq r_i + r_{i'}$. If $x_i$
and $x_{i'}$ were in the same class $V_l$, then they would be adjacent in $H$;
this is impossible, for then they could not both be members of $T_l$ (for a node
in $V_l$, membership in $T_l$ is necessary to join $F^*$). If $x_i$ and $x_{i'}$
were in different classes, then assume WLOG that $r_i < r_{i'}$. Then
$D(x_i,x_{i'}) \leq r_i + r_{i'} \leq 2 r_{i'}$, and $x_{i'}$ should not have
opened. These contradictions imply that there is at most one open facility $x_i$
for which $D(x_i,y_j) \leq r_i$.

For the rest of this lemma, then, assume that $x_i \in F^*$ is the unique open
node such that $D(x_i,y_j) \leq r_i$ (if such a $x_i$ does not exist, there is
nothing to prove). Also, let $x_{j'}$ be a minimizer for $r_x + D(x,y_j)$ so
that $\overline{r}_j = r_{j'} + D(x_{j'},y_j)$.

Now, suppose that $3 \cdot \overline{r}_j < r_i - D(x_i,y_j)$. Then
$3 \cdot r_{j'} + 3 \cdot D(x_{j'},y_j) < r_i - D(x_i,y_j)$, and so we can
conclude that (i) $3 r_{j'} < r_i$ (and $x_{j'}$ is in a lower class than $x_i$)
and (ii) $D(x_i,x_{j'}) \leq D(x_i,y_j) + D(x_{j'},y_j) \leq r_i + r_i = 2 r_i$,
which implies that $x_i$ should not have opened. This is a contradiction, and so
therefore it must be that $r_i - D(x_i,y_j) \leq 3 \overline{r}_j$. Since $x_i$
is unique (if it exists), this completes the proof.
\end{proof}

\noindent 
We are now ready to present the final result of this section.
%Theorem 2
\begin{theorem}
Algorithm \ref{alg:bipartite_facloc} (\textsc{LocateFacilities}) computes an
$O(s)$-factor approximation to
\textsc{BipartiteFacLoc} in $O(\mathcal{T}(n,s))$ rounds, where
$\mathcal{T}(n,s)$ is the running time of procedure \textsc{RulingSet}($H, s$),
called an $n$-node graph $H$.
\label{theorem:cost_approx}
\end{theorem}
\begin{proof}
Combining Lemma \ref{lemma:open_fac_dist} and Lemma \ref{lemma:open_fac_contrib}
gives
\begin{align*}
FacLoc(F^*) &= \sum\limits_{j=1}^{n_c} charge(F^*,y_j)\\
&= \sum\limits_{j=1}^{n_c} %
\left[D(F^*,y_j) + \sum_{x_i \in F^*} \max \{0, r_i - D(x_i,y_j)\}\right]\\
&\leq \sum\limits_{j=1}^{n_c} %
\left[(15 s + 15) \cdot \overline{r}_j + 3 \overline{r}_j\right]\\
&\leq (15 s + 18) \cdot \sum\limits_{j=1}^{n_c} \overline{r}_j\\
&\leq 6 \cdot (15 s + 18) \cdot OPT.
\end{align*}
The last inequality follows from the lower bound established in Lemma \ref{lemma:cost_bound}.

Also, noting that all the steps in Algorithm \ref{alg:bipartite_facloc}, except
the one that calls \textsc{RulingSet}($\cup_k H_k, s$) take a total of $O(1)$
communication rounds, we obtain the theorem.
\end{proof}

\noindent
\textbf{Note on the size of the constant.} The above analysis yields the approximation factor
$90s + 108$, which amounts to 288, since we describe how to compute a 2-ruling set.
This is obviously huge, but we have made no attempt to optimize it.
A small improvement in the size of this constant can be obtained by using multiplier $1 + \frac{1}{\sqrt{2}}$
instead of 3 in the definition of the classes $V_0, V_1, \ldots$ (see \cite{BHP12arxiv}).
For improved exposition, we use the multiplier 3.

\section{Dissemination on a Bipartite Network}
\label{section:Dissemination}

In the previous section we reduced \textsc{BipartiteFacLoc} to the problem of
computing an $s$-ruling set on a graph $H = \cup_k H_k$ defined on
facilities. Our technique for finding an $s$-ruling set involves selecting a set
$M$ of facilities at random, disseminating the induced subgraph $H[M]$ to every
client and then having each client locally compute an MIS of $H[M]$ (details
appear in Section \ref{section:2RulingSet}). A key subroutine needed to implement
this technique is one that can disseminate $H[M]$ to every client efficiently,
provided the number of edges in $H[M]$ is at most $n_f$. In Section
\ref{section:Introduction} we abstracted this problem as the
\textbf{Message Dissemination with Duplicates} (MDD) problem. In this section,
we present a randomized algorithm for MDD that runs in expected
$O(\log\log \min \{n_f, n_c\})$ communication rounds.

Recall that the difficulty in disseminating $H[M]$ is the fact that the
adjacencies in this graph are witnessed only by clients, with each adjacency
being witnessed by at least one client. However, an adjacency can be witnessed
by many clients and a client is unaware of who else has knowledge of any
particular edge. Thus, even if $H[M]$ has at most $n_f$ edges, the total number
of adjacency observations by the clients could be as large as $n_f^2$. Below we
use iterative probabilistic hashing to rapidly reduce the number of
``duplicate'' witnesses to adjacencies in $H[M]$. Once the total number of
distinct adjacency observations falls to $48 n_f$, it takes only a constant
number of additional communication rounds for the algorithm to finish
disseminating $H[M]$.
The constant ``48'' falls out easily from our analysis (Lemma \ref{lemma:number_steps}, in particular)
and we have made no attempt to optimize it in any way.

\subsection{Algorithm}

The algorithm proceeds in iterations and in each iteration a hash function is
chosen at random for hashing messages held by clients onto facilities. Denote
the universe of possible adjacency messages by $\mathcal{U}$. Since messages
represent adjacencies among facilities, $|\mathcal{U}| = \binom{n_f}{2}$.
However, it is convenient for $|\mathcal{U}|$ to be equal to $n_f^2$ and so we
extend $\mathcal{U}$ by dummy messages so that this is the case. We now define a
family $\mathcal{H}_{\mathcal{U}}$ of hash functions from $\mathcal{U}$ to
$\{1, 2, \ldots, n_f\}$ and show how to pick a function from this family,
uniformly at random. To define $\mathcal{H}_{\mathcal{U}}$, fix an ordering
$m_1, m_2, m_3, \ldots$ of the messages of $\mathcal{U}$. Partition
$\mathcal{U}$ into groups of size $n_f$, with messages
$m_1, m_2, \ldots, m_{n_f}$ as the first group, the next $n_f$ elements as the
second group, and so on. The family $\mathcal{H}_{\mathcal{U}}$ is obtained by
independently mapping each group of messages onto $(1, 2, \ldots, n_f)$ via a
cyclic permutation. For each group of $n_f$ messages in $\mathcal{U}$, there are
precisely $n_f$ such cyclic maps for it, and so a map in
$\mathcal{H}_{\mathcal{U}}$ can be selected uniformly at random by having each
facility choose a random integer in $\{1, 2, \ldots, n_f\}$ and broadcast this
choice to all clients (in the first round of an iteration). Each client then
interprets the integer received from facility $x_i$ as the image of message
$m_{(i - 1) \cdot n_f + 1}$.

In round 2, each client chooses a destination facility for each adjacency
message in its possession (note that no client possesses more than $n_f$
messages), based on the hash function chosen in round 1. For a message $m$ in
the possession of client $y_j$, $y_j$ computes the hash $h(m)$ and marks $m$ for
delivery to facility $x_{h(m)}$. In the event that more than one of $y_j$'s
messages are intended for the same recipient, $y_j$ chooses one uniformly at
random for \textit{correct} delivery, and marks the other such messages as
``leftovers.'' During the communication phase of round 2, then, client $y_j$
delivers as many messages as possible to their correct destinations; leftover
messages are delivered uniformly at random over unused communication links to
other facilities.

\begin{algo}
\textbf{Input:} A complete bipartite graph $G$, with partition $(\mathcal{F}, \mathcal{C})$;
an overlay network $H$ on $\mathcal{F}$ with $|E[H]| \leq n_f$\\
\textbf{Assumption:} For each adjacency $e'$ in $H$,
\textit{one or more} clients has knowledge of $e'$\\
\textbf{Output:} Each client should know the entire contents of $E[H]$
{\small
\begin{tabbing}
......\=a..\=b..\=c..\=d..\=e..\=f..\=g..\=h..\=i..\=j..\=k..\=l\kill
1.\>\textbf{while} $true$ \textbf{do}\\
\>\>\textbf{\emph{Start of Iteration:}}\\
2.\>\>Each client $y_j$ sends the number of distinct messages currently held,
$n_j$, to facility $x_1$.\\
3.\>\>\textbf{if} $\sum_{j=1}^{n_c} n_j \leq 48 n_f$ \textbf{then}\\
4.\>\>\>Facility $x_1$ broadcasts a \textit{break} message to each client.\\
5.\>\>\>Client $y_1$, upon receiving a \textit{break} message, broadcasts a
\textit{break} message to each facility.\\
\>\>\textbf{end-if-then}\\
6.\>\>Each facility $x_i$ broadcasts an integer in $\{1, \ldots, n_f\}$ chosen
uniformly at random; this collection of\\
\>\>\>broadcasts determines a map $h \in \mathcal{H}_{\mathcal{U}}$.\\
7.\>\>For each adjacency message $m'$ currently held, client $y_j$ maps $m'$ to
$x_{h(m')}$.\\
8.\>\>For each $i \in \{1, \ldots, n_f\}$, if
$|\{m' \text{ held by } y_j \; : \; h(m') = i\}| > 1$, client $y_j$ chooses one
message to send\\
\>\>\>to $x_i$ at random from this set and marks the others as
\textit{leftovers}.\\
9.\>\>Each client $y_j$ sends the messages chosen in Lines 7-8 to their
destinations; \textit{leftover} messages are\\
\>\>\>delivered to other facilities (for whom $y_j$ has no intended message)
in an arbitrary manner\\
\>\>\>(such that $y_j$ sends at most one message to each facility).\\
10.\>\>Each facility $x_i$ receives a collection of at most $n_c$ facility
adjacency messages; if duplicate messages\\
\>\>\>are received, $x_i$ discards all but one of them so that the messages
held by $x_i$ are distinct.\\
11.\>\>Each facility $x_i$ sends its number of distinct messages currently held,
$b_i$, to client $y_1$.\\
12.\>\>Client $y_1$ responds to each facility $x_i$ with an index
$c(i) = (\sum_{k=1}^{i-1} b_k \mod n_c)$.\\
13.\>\>Each facility $x_i$ distributes its current messages evenly to the
clients in the set\\
\>\>\>$\{y_{c(i)+1}, y_{c(i)+2}, \ldots, y_{c(i)+b_i}\}$ (where indexes are
reduced modulo $n_c$ as necessary).\\
14.\>\>Each client $y_j$ receives at most $n_f$ messages; the numbers of
messages received by any two clients\\
\>\>\>differ by at most one.\\
15.\>\>Each client discards any duplicate messages held.\\
\>\>\textbf{\emph{End of Iteration:}}\\
16.\>At this point, at most $48 n_f$ total messages remain among the $n_c$
clients; these messages may be\\
\>\>distributed evenly to the facilities in $O(1)$ communication rounds.\\
17.\>The $n_f$ facilities can now broadcast the (at most) $2 n_f$ messages to
all clients in $O(1)$ rounds.
\end{tabbing}}
\caption{\textsc{DisseminateAdjacencies}}
\label{alg:DisseminateAdjacencies}
\end{algo}

In round 3, a facility has received a collection of up to $n_c$ messages, some
of which may be duplicates of each other. After throwing away all but one copy
of any duplicates received, each facility announces to client $y_1$ the number
of (distinct) messages it has remaining. In round 4, client $y_1$ has received
from each facility its number of distinct messages, and computes for each an
index (modulo $n_c$) that allows facilities to coordinate their message
transfers in the next round. Client $y_1$ transmits the indices back to the
respective facilities in round 5.

In round 6, facilities transfer their messages back across the bipartition to
the clients, beginning at their determined index (received from client $y_1$)
and working modulo $n_c$. This guarantees that the numbers of messages received
by two clients $y_j$, $y_{j'}$ in this round can differ by no more than one.
(Although it is possible that some of these messages will ``collapse'' as
duplicates.) Clients now possess subsets of the original $n_f$ messages, and the
next iteration can begin.

\subsection{Analysis}

Algorithm \ref{alg:DisseminateAdjacencies} is proved correct by observing that
(i) the algorithm terminates only when dissemination has been completed; and
(ii) for a particular message $m'$, in any iteration, there is a nonzero
probability that all clients holding a copy of $m'$ will deliver $m'$ correctly,
after which there will never be more than one copy of $m'$ (until all messages
are broadcast to all clients at the end of the algorithm). The running time
analysis of Algorithm \ref{alg:DisseminateAdjacencies} starts with two lemmas
that follow from our choice of the probabilistic hash function.

%Lemma 11
\begin{lemma}
Suppose that, at the beginning of an iteration, client $y_j$ possesses a
collection $S_j$ of messages, with $|S_j| = n_j$. Let $E_{i,j}$ be the event
that at least one message in $S_j$ hashes to facility $x_i$. Then the
probability of $E_{i,j}$ (conditioned on all previous iterations) is bounded
below by $1 - e^{-n_j / n_f}$.
\label{lemma:bound_hash}
\end{lemma}
\begin{proof}
Let $R_{j,k}$ be the intersection of $S_j$ with the $k$th group of
$\mathcal{U}$, so that we have the partition
$S_j = R_{j,1} + R_{j,2} + \ldots + R_{j,n_f}$. Let $E_{i,j,k}$ be the event
that some message in $R_{j,k}$ hashes to $x_i$, so that
$E_{i,j} = \cup_{k=1}^{n_f} E_{i,j,k}$. Due to the nature of
$\mathcal{H}_{\mathcal{U}}$ as maps on $R_{j,k}$ (a uniformly distributed
collection of cyclic injections), the probability of $E_{i,j,k}$,
$\mathbf{P}(E_{i,j,k})$, is precisely $|R_{j,k}| / n_f$. Therefore, the
probability of the complement of $E_{i,j,k}$,\\[1mm]
$\mathbf{P}(\overline{E_{i,j,k}})$, is $1 - |R_{j,k}| / n_f$. Using the
inequality $1 - n x \leq (1 - x)^n$ (for $x \in [0, 1]$), we can bound
$\mathbf{P}(\overline{E_{i,j,k}})$\\[1mm]
above by $(1 - \frac{1}{n_f})^{|R_{j,k}|}$.

Next, since the actions of $h$ (from $\mathcal{H}_{\mathcal{U}}$) on each
$R_{j,k}$ are chosen independently, the events $\{E_{i,j,k}\}_{k=1}^{n_f}$ are
(mutually) independent. Therefore, we have
\begin{align*}
\mathbf{P}\left(\overline{E_{i,j}}\right) &= %
\mathbf{P}\left(\overline{\cup_{k=1}^{n_f} E_{i,j,k}}\right)\\[1mm]
&= \mathbf{P}\left(\bigcap\limits_{k=1}^{n_f} \overline{E_{i,j,k}}\right)\\
&= \prod\limits_{k=1}^{n_f} \mathbf{P}\left(\overline{E_{i,j,k}}\right)\\
&\leq \prod\limits_{k=1}^{n_f} \left(1 - \frac{1}{n_f}\right)^{|R_{j,k}|}\\[1mm]
&= \left(1 - \frac{1}{n_f}\right)^{|R_{j,1}| + |R_{j,2}| + \ldots + |R_{j,n_f}|}\\[1mm]
&= \left(1 - \frac{1}{n_f}\right)^{n_j}\\
\end{align*}
Using the inequality $1 + x \leq e^x$ (for all $x$), we can then bound
$\mathbf{P}(\overline{E_{i,j}})$ above by
$(e^{-\frac{1}{n_f}})^{n_j} = e^{-\frac{n_j}{n_f}}$. Thus we have
$\mathbf{P}(E_{i,j}) \geq 1 - e^{-\frac{n_j}{n_f}}$.
\end{proof}

%Lemma 12
\begin{lemma}
Suppose that, at the beginning of an iteration, client $y_j$ possesses a
collection $S_j$ of messages, with $|S_j| = n_j$. Let $M_j \subseteq S_j$ be the
subset of messages that are correctly delivered by client $y_j$ in the present
iteration. Then the expected value of $|M_j|$ (conditioned on previous
iterations) is bounded below by $n_j - \frac{n_j^2}{2 n_f}$.
\label{lemma:expect_deliver}
\end{lemma}
\begin{proof}
Let $M_{j,i} \subseteq M_j$ be the subset of messages correctly delivered by
$y_j$ to $x_i$ (in the present iteration), so that we have the partition
$M_j = M_{j,1} + M_{j,2} + \ldots + M_{j,n_f}$. Observe that $|M_{j,i}|$ is $1$
if at least one message in $S_j$ hashes to $x_i$, and $0$ otherwise. Therefore
$|M_{j,i}|$ is equal to $1_{E_{i,j}}$, the indicator random variable for event
$E_{i,j}$, and so the expected value of $|M_{j,i}|$, $\mathbf{P}(|M_{j,i}|)$, is
equal to $\mathbf{P}(E_{i,j})$. By linearity of expectation, we have
$$\mathbf{E}(|M_j|) = \sum\limits_{i=1}^{n_f} \mathbf{E}(|M_{j,i}|)
= \sum\limits_{i=1}^{n_f} \mathbf{P}(E_{i,j})
\geq n_f \cdot \left(1 - e^{-\frac{n_j}{n_f}}\right)$$
Now, the function $1 - e^{-x}$ has the (alternating) Taylor series
$$1 - e^{-x} = 1 - \left(1 - x + \frac{x^2}{2} - \frac{x^3}{6} + \ldots\right) = x - \frac{x^2}{2} + \frac{x^3}{6} - \ldots$$
Since $0 \leq n_j \leq n_f$, with $x = n_j / n_f$ we have an alternating series
with terms of decreasing magnitude; thus we have the lower bound
$$\mathbf{E}(|M_j|) \geq n_f \cdot \left(1 - e^{-\frac{n_j}{n_f}}\right)
\geq n_f \cdot \left(\frac{n_j}{n_f} - \frac{n_j^2}{2 n_f^2}\right)
= n_j - \frac{n_j^2}{2 n_f}.$$
\end{proof}

\noindent By Lemma \ref{lemma:expect_deliver}, the number of incorrectly
delivered messages in $S_j$ is bounded above (in expectation) by
$\frac{n_j^2}{2n_f}$. Informally speaking, this implies that the sequence
$n_f, \frac{n_f}{2}, \frac{n_f}{2^3}, \frac{n_f}{2^7}, \ldots$ bounds from above
the number of incorrectly delivered messages (in expectation) in each iteration.
This doubly-exponential rate of decrease in the number of undelivered messages
leads to the expected-doubly-logarithmic running time of the algorithm.

We now step out of the context of a single client and consider the progress of the
algorithm on the whole. Using Lemma \ref{lemma:expect_deliver}, we derive the
following recurrence for the expected total number of messages held by all
clients at the beginning of each iteration.

%Lemma 13
\begin{lemma}
Suppose that the algorithm is at the beginning of iteration $I$, $I \geq 2$, and
let $T_I$ be the total number of messages held by all clients
(i.e. $T_I = \sum_{j=1}^{n_c} n_j(I)$, where $n_j(I)$ is the number of messages
held by client $y_j$ at the beginning of iteration $I$). Then the conditional
expectation of $T_{I+1}$ given $T_I$, $\mathbf{E}(T_{I+1} \mid T_I)$, satisfies
\begin{displaymath}
\mathbf{E}(T_{I+1} \mid T_I) \leq
\begin{cases}
n_f + \frac{(T_I + n_c)^2}{2 n_f \cdot n_c} & \text{ if } T_I > n_c\\[1mm]
n_f + \frac{T_I}{2 n_f} & \text{ if } T_I \leq n_c\\
\end{cases}
\end{displaymath}
\label{lemma:expect_progress}
\end{lemma}
\begin{proof}
Since $I \geq 2$, at the beginning of iteration $I$, the $T_I$ messages are
evenly spread among all clients; the numbers $n_j$, $n_{j'}$ of messages held by
two distinct clients $y_j$, $y_{j'}$ differ by no more than $1$. Therefore,
$n_j \leq T_I / n_c + 1$ for all $j$. As well, if $T_I \leq n_c$, then $n_j = 1$
for $T_I$ values of $j$, and $0$ otherwise.

The number of messages remaining after iteration $I$ is bounded above by the
number of messages not correctly delivered during iteration $I$, plus $n_f$
(for each collection of identical messages that ``collapse'' at a given facility
$x_i$, one such message is kept and passed back to some client). Therefore,
$T_{I+1} \leq n_f + \sum_{j=1}^{n_c} |S_j \setminus M_j|$ (where $S_j$ is as
defined in Lemma \ref{lemma:bound_hash}, and $M_j$ as in Lemma
\ref{lemma:expect_deliver}). We then have
\begin{align*}
\mathbf{E}(T_{I+1} \mid T_I) &\leq %
n_f + \sum\limits_{j=1}^{n_c} \mathbf{E}(|S_j \setminus M_j| \mid T_I)\\
&\leq n_f + \sum\limits_{j=1}^{n_c} \frac{n_j^2}{2 n_f}\\
&\leq n_f + \sum\limits_{j=1}^{n_c} \frac{(\frac{T_I}{n_c} + 1)^2}{2 n_f}\\
&= n_f + n_c \cdot \frac{(\frac{T_I}{n_c} + 1)^2}{2 n_f}\\
&= n_f + \frac{(T_I + n_c)^2}{2 n_f \cdot n_c}\\
\end{align*}
If $T_I \leq n_c$, we get also that
\begin{align*}
\mathbf{E}(T_{I+1} \mid T_I) &\leq %
n_f + \sum\limits_{j=1}^{n_c} \frac{n_j^2}{2 n_f}\\
&= n_f + \frac{T_I}{2 n_f}
\end{align*}
\end{proof}

\noindent
We now define a sequence of variables $t_i$ (via the recurrence below) that bounds from above the expected behavior of the
sequence of $T_I$'s established in the previous lemma.
Let $t_1 = n_f \cdot \min \{n_f, n_c\}$, $t_i = \frac{1}{2} t_{i-1}$
for $2 \leq i \leq 5$, and for $i > 5$, define $t_i$ by
\begin{displaymath}
t_i =
\begin{cases}
2 n_f + \frac{(t_{i-1} + n_c)^2}{n_f \cdot n_c} & \text{ if } t_{i-1} > n_c\\[1mm]
2 n_f + \frac{t_{i-1}}{n_f} & \text{ if } t_{i-1} \leq n_c\\
\end{cases}
\end{displaymath}
The following lemma establishes that the $t_i$'s fall rapidly.
%Solving this probabilistic recurrence leads to Theorem
%\ref{theorem:dissemination_rounds}.

%Lemma 14
\begin{lemma}
The smallest index $i$ for which $t_i \leq 48 n_f$ is
at most $\log \log \min \{n_f, n_c\} + 2$.
\label{lemma:number_steps}
\end{lemma}
\begin{proof}
Equivalently, we concern ourselves with $t_{i}' = \frac{4 t_i}{n_f}$ and
determine the number of rounds required before this quantity becomes bounded
above by 192.

If $n_c \le 48$ or $n_f \le 48$, then $t_{1}' = 4 \min \{n_f, n_c\} \le 192$ and we are done.
So we assume that both $n_c$ and $n_f$ are greater than 48.
Now rewrite the recursion for $t_i$ ($i > 5$)
as
\begin{displaymath}
t_i \leq
\begin{cases}
2 n_f + \frac{(2 t_{i-1})^2}{n_f \cdot n_c} & \text{ if } t_{i-1} > n_c\\[1mm]
2 n_f + \frac{t_{i-1}}{n_f} & \text{ if } t_{i-1} \leq n_c\\
\end{cases}
\end{displaymath}
Correspondingly, write the recurrence for $t_i'$ as follows.
\begin{displaymath}
t_{i}' \leq
\begin{cases}
8 + \frac{(t_{i-1}')^2}{n_c} & \text{ if } t_{i-1}' > \frac{4 n_c}{n_f}\\[1mm]
8 + \frac{t_{i-1}'}{n_f} & \text{ if } t_{i-1}' \leq \frac{4 n_c}{n_f}\\
\end{cases}
\end{displaymath}
We now prove the following claim by induction: for each $i = 5, 6, \ldots$,
$t_i' \le \min\{4n_f, n_c/4\}$.
The base case concerns $t_5'$.
Since $t_5' = \frac{1}{4} \min\{n_f, n_c\}$, the claim is clearly true for $i = 5$.
Assuming that the claim is true for $t_i'$, we now consider $t_{i+1}'$.
First note that since $t_i' \le \min\{4n_f, n_c/4\}$,
it follows that $(t_i')^2 \le n_c \cdot n_f$ and $(t_i')^2 \le n_c^2/16\}$.
Now we consider the two possible cases of the recursion to get a bound on $t_{i+1}'$.
\begin{enumerate}
\item If $t_i' > 4n_c/n_f$, then $t_{i+1}' \le 8 + \frac{n_c n_f}{n_c} = 8 + n_f$.
Since $n_f \ge 48$, we have that $8 + n_f \le 4 n_f$ and therefore $t_{i+1}' \le 4 n_f$.
Similarly, if $t_i' > 4n_c/n_f$, we also have that $t_{i+1}' \le 8 + \frac{n_c}{16}$.
Again, since $n_c \ge 48$, it follows that $8 + \frac{n_c}{16} \le \frac{n_c}{4}$.
Hence, $t_{i+1}' \le n_c/4$.

\item If $t_i' \le 4n_c/n_f$, then $t_{i+1}' \le 8 + 4n_f/n_f = 12$.
Since both $n_f$ and $n_c$ are greater than 48, $12 < \min\{4n_f, n_c/4\}$.
Hence, $t_{i+1}' < \min\{4n_f, n_c/4\}$.
\end{enumerate}

To finish the proof, we now consider two cases.
\begin{description}
\item[Case 1: $n_c \ge n_f^2$.] 
In this case, $4n_c/n_f \ge 4n_f$.
According to the inductive claim proved above, $t_i' \le 4n_f$ for all $i = 5, 6, \ldots$.
Therefore, $t_5' \le 4n_c/n_f$ and Case 2 of the recurrence applies and yields 
$t_6' \le 8 + 4n_f/n_f = 12$, completing the proof. 

\item[Case 2: $n_c < n_f^2$.] 
For notational convenience, let us use $R$ to denote $c \cdot \log\log \min\{n_f, n_c\}$.
For $i = 5, 6, \ldots, R-1$, we assume that $t_i' > 4n_c/n_f$.
Otherwise, Case 2 of the recurrence applies and $t_{i+1}' \le 8 + 4n_f/n_f = 12$ and we are done.
Also, for $i = 5, 6, \ldots, R-1$, we assume that $(t_i')^2/n_c > 4$.
Otherwise, $t_{i+1}' \le 8 + 4 = 12$ (using Case 1 of the recurrence) and we are done.

Now define $t_{i}'' = 3 t_{i}'$;
we bound $t_{i}''$ above by a sequence that falls at a double-exponential rate. 
Given that $t_i' > 4n_c/n_f$ for $i = 5, 6, \ldots, R-1$,
we see that 
$$3t_{i+1}' \le 24 + \frac{3 (t_i')^2}{n_c}.$$
Furthermore, given that $(t_i')^2/n_c > 4$ for $i = 5, 6, \ldots, R-1$,
we see that
$$3t_{i+1}' \le \frac{3 (t_i')^2}{n_c} + \frac{3 (t_i')^2}{n_c} = \frac{9 (t_i')^2}{n_c} = \frac{(t_i'')^2}{n_c}.$$
%[Note that, if $t_{i-1}' \leq 12 \leq n_f$ and $n_c \geq 36$, then
%$t_{i}' \leq 8 + 4 = 12$, and so once $t_{i-1}'$ falls below a certain constant,
%it will never again exceed it.]
%\vspace{3mm}
%
%\noindent\textit{Case I.} Suppose $n_f^2 \leq n_c$. Then
%$t_{5}' \leq 4 n_f \leq \frac{4 n_c}{n_f}$ and so
%$t_{6}' = 8 + \frac{t_{4}'}{n_f} \leq 8 + 4 = 12$.
%\vspace{3mm}
%\noindent\textit{Case II.} If $n_f^2 > n_c$, then define $t_{i}'' = 3 t_{i}'$;
%we bound $t_{i}''$ above by a sequence containing a double-exponential. Assuming
%the previous conditions on $n_f$ and $n_c$ (that they are greater than certain
%constants), we may then assume without loss of generality that
%$\frac{(t_{i-1}')^2}{n_c} \geq 4$, and so
Thus $t_{i+1}'' \leq (t_{i}'')^2/n_c$. Now, $t_{5}'' \leq \frac{3}{4} n_c$, and
so by induction, $t_{5+j}'' \leq \left(\frac{3}{4}\right)^{2^j} \cdot n_c$. Thus
the smallest $j$ for which $t_{5+j}'' \le 192$ is  at most $2 + \log \log n_c$, which in
this case is also $2 + \log \log \min \{n_f, n_c\}$.
\end{description}
\end{proof}

%Lemma 15
\begin{lemma}
For $i > 5$, if $T_I \leq t_i$, then the conditional probability (given
iterations $1$ through $I-1$) of the event that $T_{I+1} \leq t_{i+1}$ is
bounded below by $\frac{1}{2}$.
\label{lemma:bound_prob}
\end{lemma}
\begin{proof}
If $i > 5$ and $T_I \leq t_i$, then by Lemma \ref{lemma:expect_progress},
$\mathbf{E}(T_{I+1}) \leq \frac{1}{2} t_{i+1}$. Therefore, by Markov's
inequality, $\mathbf{P}(T_{I+1} > t_{i+1}) \leq \frac{1}{2}$ and
$\mathbf{P}(T_{I+1} \leq t_{i+1}) \geq \frac{1}{2}$.
\end{proof}

%Theorem 3
\begin{theorem}
Algorithm \ref{alg:DisseminateAdjacencies} solves the dissemination problem in
$O(\log \log \min \{n_f, n_c\})$ rounds in expectation.
\label{theorem:dissemination_rounds}
\end{theorem}
\begin{proof}
Let $\tau_i = \min_{\{I \; : \; T_I \leq t_{i+1} \}} \{I\} -
\min_{\{I \; : \; T_I \leq t_i\}} \{I\}$. Conceptually, $\tau_i$ is the number
of rounds necessary for the total number of messages remaining to decrease from
$t_i$ to $t_{i+1}$. Thus, the running time of Algorithm
\ref{alg:DisseminateAdjacencies} is $O(1) +
\sum_{i=1}^{O(\log \log \min\{n_f, n_c\})} \tau_i$. By linearity of
expectation, the expected running time is then
$O(1) + \sum_{i=1}^{O(\log \log \min \{n_f, n_c\})} \mathbf{E}(\tau_i)$. By Lemma
\ref{lemma:bound_prob}, if $T_I \leq t_i$, then regardless of past history,
there is at least a probability-$\frac{1}{2}$ chance that $T_{I+1}$ will be less
than $t_{i+1}$. It follows that $\tau_i$ is dominated by an Exp($\frac{1}{2}$)
(exponential) random variable, and so $\mathbf{E}(\tau_i) \leq 2$. Therefore the
expected running time of Algorithm \ref{alg:DisseminateAdjacencies} is
$O(\log \log \min \{n_f, n_c\})$.
\end{proof}

\section{Computing a 2-Ruling Set of Facilities}
\label{section:2RulingSet}

In this section, we show how to efficiently compute a 2-ruling set on the graph
$H$ (with vertex set $\mathcal{F}$) constructed in Algorithm
\ref{alg:bipartite_facloc} (\textsc{LocateFacilities}). Our algorithm (called \textsc{Facility2RulingSet}
and described as Algorithm \ref{alg:Bipartite2RulingSet}) computes a $2$-ruling
set in $H$ by performing \textit{iterations} of a procedure that combines
randomized and deterministic sparsification steps. In each iteration, each
facility chooses (independently) to join the \textit{candidate set} $M$ with
probability $p$. Two neighbors in $H$ may both have chosen to join $M$, so $M$
may not be independent in $H$. We would therefore like to select an MIS of the
graph induced by $M$, $H[M]$. In order to do this, the algorithm attempts to
communicate all known adjacencies in $H[M]$ to every client in the network, so
that each client may (deterministically) compute the same MIS. The algorithm relies on
Algorithm \textsc{DisseminateAdjacencies}
(Algorithm \ref{alg:DisseminateAdjacencies}) developed in Section
\ref{section:Dissemination} to perform this communication.

\begin{algo}
\textbf{Input:} Complete bipartite graph $G$ with partition $(\mathcal{F}, \mathcal{C})$ and $H$, an overlay
network on $\mathcal{F}$.\\
\textbf{Output:} A $2$-ruling set $T$ of $H$
{\small
\begin{tabbing}
......\=a..\=b..\=c..\=d..\=e..\=f..\=g..\=h..\=i..\=j..\=k..\=l\kill
1.\>$i := 1$; $p := p_1 = \frac{1}{8 \cdot n_f^{1 / 2}}$; $T := \emptyset$\\
2.\>\textbf{while} $|E(H)| > 0$ \textbf{do}\\
\>\>\textbf{\emph{Start of Iteration:}}\\
3.\>\>$M := \emptyset$\\
4.\>\>Each facility $x$ joins $M$ with a probability $p$.\\
5.\>\>Run Algorithm \textsc{DisseminateAdjacencies} for
$7 \log \log \min \{n_f, n_c\}$ iterations to communicate\\
\>\>\>the edges in $H[M]$ to all clients in the network.\\
6.\>\>\textbf{if} Algorithm \textsc{DisseminateAdjacencies} completes within
the allotted number of iterations \textbf{then}\\
7.\>\>\>Each client computes the same MIS $L$ on $M$ using a deterministic
algorithm.\\
8.\>\>\>$T := T \cup L$\\
9.\>\>\>Remove $M \cup N(M)$ from $H$.\\
10.\>\>\>$i := i + 1$; $p := p_i = \frac{1}{8 \cdot n_f^{2^{-i}}}$\\
11.\>\>\textbf{else}\\
12.\>\>\>$i := i - 1$; $p := p_i = \frac{1}{8 \cdot n_f^{2^{-i}}}$\\
13.\>\>\textbf{if} $|E(H)| = 0$ \textbf{then break};\\
\>\>\textbf{\emph{End of Iteration:}}\\
14.\> Output $T$.
\end{tabbing}}
\caption{\textsc{Facility2RulingSet}}
\label{alg:Bipartite2RulingSet}
\end{algo}

For Algorithm \textsc{DisseminateAdjacencies} to terminate quickly, we require
that the number of edges in $H[M]$ be $O(n_f)$. This requires the probability
$p$ to be chosen carefully as a function of $n_f$ and the number of edges in
$H$. Due to the lack of aggregated information, nodes of the network do not
generally know the number of edges in $H$ and thus the choice of $p$ may be
``incorrect'' in certain iterations. To deal with the possibility that $p$ may
be too large (and hence $H[M]$ may have too many edges), the dissemination
procedure is not allowed to run indefinitely -- rather, it is cut off after
$7 \log \log \min \{n_f, n_c\}$ iterations of disseminating hashing. If
dissemination was successful, i.e. the subroutine completed prior to the cutoff,
then each client receives complete information about the adjacencies in $H[M]$, and thus each is able
to compute the same MIS in $H[M]$. Also, if dissemination was successful, then $M$
and its neighborhood, $N(M)$, are removed from $H$ and the next iteration is run
with a larger probability $p$. On the other hand, if dissemination was
unsuccessful, the current iteration of \textsc{Facility2RulingSet} is terminated
and the next iteration is run with a smaller probability $p$ (to make success
more likely the next time).

To analyze the progress of the algorithm, we define two notions --
\textit{states} and \textit{levels}. For the remainder of this section, we use
the term \textit{state} (of the algorithm) to refer to the current probability
value $p$. The probability $p$ can take on values
$\left(\frac{1}{8 \cdot n_f^{2^{-i}}}\right)$ for
$i = 0, 1, \ldots, \Theta(\log \log n_f)$. We use the term \textit{level} to
refer to the progress made up until the current iteration.
Specifically, the $j$th level
$L_j$, for $j = 0, 1, \ldots, \Theta(\log \log n_f)$, is defined as having been
reached when the number of facility adjacencies remaining in $H$ becomes less
than or equal to $l_j = 8 \cdot n_f^{1 + 2^{-j}}$. In addition, we define one
special level $L_{\ast}$ as the level in which no facility adjacencies remain.
These values for the states and levels are chosen so that, once level $L_i$ has
been reached, one iteration run in state $i + 1$ has at least a
probability-$\frac{1}{2}$ chance of advancing progress to level $L_{i+1}$.

\subsection{Analysis}

It is easy to verify that the set $T$ computed by Algorithm
\ref{alg:Bipartite2RulingSet} (\textsc{Facility2RulingSet}) is a $2$-ruling set and we now turn our attention
to the expected running time of this algorithm. The
algorithm halts exactly when level $L_{\ast}$ is reached (this termination
condition is detected in Line 15), and so it suffices to bound the expected
number of rounds necessary for progress (removal of edges from $H$) to reach
level $L_{\ast}$. The following lemmas show that quick progress is made when the
probability $p$ matches the level of progress made thus far.

%Lemma 16
\begin{lemma}
Suppose $|E(H)| \leq l_i$ (progress has reached level $L_i$) and in this
situation one iteration is run in state $i + 1$ (with $p = p_{i+1}$). Then in
this iteration, the probability that Algorithm \textsc{DisseminateAdjacencies}
succeeds is at least $\frac{3}{4}$.
\label{lemma:progress_dissemination}
\end{lemma}
\begin{proof}
Let $a$ refer to the number of adjacencies (edges) in $H[M]$. With $p = p_{i+1}$,
$\mathbf{E}(a) = |E(H)| \cdot p_{i+1}^2 \leq l_i \cdot p_{i+1}^2$. Plugging
the values of $l_i$ and $p_{i+1}$ into this bound, we see that
$\mathbf{E}(a) \leq \frac{n_f}{8}$. By Markov's inequality,
$\mathbf{P}(a > n_f) \leq \frac{\mathbf{E}(a)}{n_f} = \frac{1}{8}$.

Let $T_d$ be the number of iterations that dissemination \textit{would} run for
if it were allowed to run to completion in this iteration. (Recall that,
regardless of $T_d$, we always terminate dissemination after
$7 \log \log \min \{n_f, n_c\}$ iterations.) By Theorem
\ref{theorem:dissemination_rounds},
$\mathbf{E}(T_d \mid a \leq n_f) \leq \log \log \min \{n_f, n_c\}$. Therefore,
$\mathbf{P}(T_d > 7 \log \log \min \{n_f, n_c\} \mid a \leq n_f)$ is bounded
above by $\frac{1}{7}$ (again, using Markov's inequality). If $E_c$ is the event
that $a > n_f$, and $E_T$ is the event that
$T_d > 7 \log \log \min \{n_f, n_c\}$, then we can bound $\mathbf{P}(E_T)$ above
by
\begin{align*}
\mathbf{P}(E_T) &\leq \mathbf{P}(E_c \cup E_T)\\[1mm]
&= \mathbf{P}(E_c) + \mathbf{P}(E_T \cap \overline{E_c})\\[1mm]
&= \mathbf{P}(E_c) + %
\mathbf{P}(E_T \mid \overline{E_c}) \cdot \mathbf{P}(\overline{E_c})\\
&\leq \frac{1}{8} + \mathbf{P}(E_T \mid \overline{E_c}) \cdot \frac{7}{8}\\
&\leq \frac{1}{8} + \frac{1}{7} \cdot \frac{7}{8}\\
&= \frac{1}{4}\\
\end{align*}
So with probability at least $\frac{3}{4}$, dissemination succeeds (completes in
the time allotted).
\end{proof}

%Lemma 17
\begin{lemma}
Suppose $|E(H)| \leq l_i$ (progress has reached level $L_i$). Then, after one
iteration run in state $i + 1$ (with $p = p_{i+1}$), the probability that level
$L_{i+1}$ will be reached (where $|E(H)| \leq l_{i+1}$) is at least
$\frac{1}{2}$.
\label{lemma:progress_pi}
\end{lemma}
\begin{proof}
In the present iteration, run with $p = p_{i+1}$, we first ignore the success or
failure of dissemination (within $7 \log \log n_f$ iterations of hashing), and
assume instead that dissemination runs as long as necessary to succeed.
Consider, in this modified scenario, the expected number of edges that will
remain in $H$. The number of edges can be calculated as twice the sum of
degrees, and we can bound the expected degree in $H$ of a facility $x$ above by
the current degree of $x$ multiplied by the probability that $x$ remains active.
The probability that $x$ remains active is at most
$(1 - p_{i+1})^{\mathrm{deg}_H(x)}$ (the probability that no neighbor of $x$
becomes a candidate). In turn, this quantity is less than or equal to
$e^{-p_{i+1} \cdot \mathrm{deg}_H(x)}$. Thus, if $m$ refers to the number of
edges remaining in $H$ after the present iteration (again, with dissemination
running to completion), we have
\begin{align*}
\mathbf{E}(m) &= \frac{1}{2} \sum\limits_{x \in \mathcal{F}} %
\mathrm{deg}(x) \cdot e^{-p_{i+1} \cdot \mathrm{deg}(x)}\\
&= \frac{1}{2 p_{i+1}} \sum\limits_{x \in \mathcal{F}} %
p_{i+1} \cdot \mathrm{deg}(x) \cdot e^{-p_{i+1} \cdot \mathrm{deg}(x)}\\
&\leq \frac{1}{2 p_{i+1}} \sum\limits_{x \in \mathcal{F}} \frac{1}{e}\\
&= \frac{1}{2 e \cdot p_{i+1}} \cdot n_f\\
&= \frac{1}{2 e} \cdot 8 n^{\frac{1}{2^{i+1}}} \cdot n_f\\
&= \frac{1}{2 e} \cdot l_{i+1}\\
\end{align*}
The inequality in the above calculation (Line 3) follows from the fact that $x \cdot e^{-x} \le e^{-1}$ for all
real $x$.
Since the unconditional expected value satisfies
$\mathbf{E}(m) \leq \frac{l_{i+1}}{2 e}$, the conditional expectation
$\mathbf{E}(m \mid \overline{E_T})$ is bounded above by
$\frac{4}{3} \cdot \mathbf{E}(m) = \frac{2}{3 e} \cdot l_{i+1}$. Recall the
definition of the event $E_T$ from the proof of Lemma
\ref{lemma:progress_dissemination}. Therefore, using Markov's inequality, the
probability that $m > l_{i+1}$ given $\overline{E_T}$ is no greater than
$\left(\frac{2 l_{i+1}}{3 e}\right) / l_{i+1} = \frac{2}{3 e} < \frac{1}{3}$.

Thus we have $\mathbf{P}(m \leq l_{i+1} \mid \overline{E_T}) > \frac{2}{3}$, and
using $\mathbf{P}(\overline{E_T}) \geq \frac{3}{4}$
(from Lemma \ref{lemma:progress_dissemination}),
\[\mathbf{P}(m \leq l_{i+1}) \geq \mathbf{P}(\{m \leq l_{i+1}\} \cap \overline{E_T})
= \mathbf{P}(m \leq l_{i+1} \mid \overline{E_T}) \cdot \mathbf{P}(\overline{E_T})
> \frac{2}{3} \cdot \frac{3}{4} = \frac{1}{2}\]
\end{proof}

\noindent
Thus, once level $L_i$ has been reached, we can expect that only a constant
number of iterations run in state $i + 1$ would be required to reach level
$L_{i+1}$. Therefore, the question is, ``How many iterations of the algorithm
are required to execute state $i + 1$ enough times?'' To answer this question,
we abstract the algorithm as a stochastic process that can be modeled as a
(non-Markov) simple random walk on the integers
$0, 1, 2, \ldots, \Theta(\log \log n_f)$ with the extra property that, whenever
the random walk arrives at state $i + 1$, a (fair) coin is flipped. We place a
bound on the expected number of steps before this coin toss comes up heads.

First, consider the return time to state $i + 1$. In order to prove that the
expected number of iterations (steps) necessary before either
$|E(H)| \leq l_{i+1}$ or $p = p_{i + 1}$ is $O(\log \log n_f)$, we consider two
regimes -- $p > p_{i + 1}$ and $p < p_{i + 1}$. When $p$ is large (in the regime
consisting of probability states intended for fewer edges than currently remain
in $H$), it is likely that a single iteration of Algorithm
\ref{alg:Bipartite2RulingSet} will generate a large number of adjacencies
between candidate facilities. Thus, dissemination will likely not complete
before ``timing out,'' and it is likely that $p$ will be decreased prior to the
next iteration. Conversely, when $p$ is small (in the regime consisting of
probability states intended for more edges than currently remain in $H$), a
single iteration of Algorithm \ref{alg:Bipartite2RulingSet} will likely generate
fewer than $n_f$ adjacencies between candidate facilities, and thus it is likely
that dissemination will complete before ``timing out.'' In this case, $p$ will
advance prior to the next iteration. This analysis is accomplished in the
following lemmas and leads to the subsequent theorem.

%Lemma 18
\begin{lemma}
Consider a simple random walk on the integers $[0, i]$ with transition
probabilities $\{p_{j,k}\}$ satisfying $p_{j,j+1} = \frac{3}{4}$
($j = 0, \ldots, i - 1$), $p_{j,j-1} = \frac{1}{4}$, ($j = 1, \ldots, i$),
$p_{i,i} = \frac{3}{4}$, and $p_{0,0} = \frac{1}{4}$. For such a random walk
beginning at $0$, the expected hitting time of $i$ is $O(i)$.
\label{lemma:hitting_time}
\end{lemma}
\begin{proof}
This is an exercise in probability; see \cite{MitzenmacherBook}.
\end{proof}

%Lemma 19
\begin{lemma}
When $j \leq i$, the expected number of iterations required before returning to
state $i + 1$ is $O(\log \log n_f)$.
\label{lemma:return_below}
\end{lemma}
\begin{proof}
By Lemma \ref{lemma:hitting_time}, it suffices to show that when $j < i$, the
probability of successful dissemination in state $j$ is at least $\frac{3}{4}$.
By the proof of Lemma \ref{lemma:progress_dissemination}, this would be true
were the current iteration run with $p = p_i$. Since $p_j < p_i$, the
probability of successful dissemination is greater in state $j$ then in state
$i$, and the lemma follows.
\end{proof}

%Lemma 20
\begin{lemma}
When $j > i$, the expected number of iterations required before returning to
state $i + 1$ or advancing to at least level $L_{i+1}$ is $O(\log \log n_f)$.
\label{lemma:return_above}
\end{lemma}
\begin{proof}
By Lemma \ref{lemma:hitting_time}, it suffices to show that when $j > i$, the
probability of either unsuccessful dissemination in state $j$ or progression to
level $L_j$ is at least $\frac{3}{4}$. Therefore, consider the modified scenario
where dissemination is always run to completion; we will show that the
probability of progression to level $L_j$ in this scenario is at least
$\frac{3}{4}$.

Recall from the proof of Lemma \ref{lemma:progress_pi} that, if $m$ refers to
the number of edges remaining in $H$ after the present iteration
(with dissemination run to completion), we have
\begin{align*}
\mathbf{E}(m) &= \frac{1}{2} \sum\limits_{x \in \mathcal{F}} %
\mathrm{deg}(x) \cdot e^{-p_{j} \cdot \mathrm{deg}(x)}\\
&= \frac{1}{2 p_{j}} \sum\limits_{x \in \mathcal{F}} %
p_{j} \cdot \mathrm{deg}(x) \cdot e^{-p_{j} \cdot \mathrm{deg}(x)}\\
&= \frac{1}{2 p_{j}} \sum\limits_{x \in \mathcal{F}} \frac{1}{e}\\
&= \frac{1}{2 e \cdot p_{j}} \cdot n_f\\
&= \frac{1}{2 e} \cdot 8 n^{\frac{1}{2^j}} \cdot n_f\\
&= \frac{1}{2 e} \cdot l_j\\
\end{align*}
Thus, by Markov's inequality,
$\mathbf{P}(m > l_j) \leq \frac{1}{2 e} < \frac{1}{4}$, and so the probability
of progression to level $L_j$ (when dissemination is allowed to run to
completion) is at least $\frac{3}{4}$, which completes the proof of the lemma.
\end{proof}

%Lemma 21
\begin{lemma}
Suppose that Algorithm \ref{alg:Bipartite2RulingSet} has reached level $L_i$,
and let $T_{i+1}$ be a random variable representing the number of iterations
necessary before reaching level $L_{i+1}$. Then
$\mathbf{E}(T_{i+1}) = O(\log \log n_f)$.
\label{lemma:expect_time_level}
\end{lemma}
\begin{proof}
Fix $i$ and let $N$ be the number of returns to state $i + 1$ prior to
progressing to level $L_{i+1}$. Let $S_k$ be the number of iterations run
between return $k - 1$ and $k$ to state $i + 1$, so that
$T_{i+1} = \sum_{k=1}^N S_k$. In the random walk abstraction of the algorithm,
$N$ depends only on a series of coin flips which are themselves independent of
all other history. Each $S_k$ is also independent of any coin flip and so also
of $N$ (again we emphasize that this is only true when abstracting the algorithm
to a random walk).

Now, by conditioning on $N = t$, we see that
\begin{align*}
\mathbf{E}(T_{i+1}) &= \sum\limits_{t=0}^{\infty} \mathbf{E}(T_{i+1} \mid N = t)\\
&= \sum\limits_{t=0}^{\infty} \mathbf{E}(\sum\limits_{k=1}^N S_k \mid N = t) \cdot \mathbf{P}(N = t)\\
&= \sum\limits_{t=0}^{\infty} \mathbf{E}(\sum\limits_{k=1}^t S_k \mid N = t) \cdot \mathbf{P}(N = t)\\
&= \sum\limits_{t=0}^{\infty} \mathbf{E}(\sum\limits_{k=1}^t S_k) \cdot \mathbf{P}(N = t)\\
&= \sum\limits_{t=0}^{\infty} \sum\limits_{k=1}^t \mathbf{E}(S_k) \cdot \mathbf{P}(N = t)\\
&\leq \sum\limits_{t=0}^{\infty} \sum\limits_{k=1}^t O(\log \log n_f) \cdot \mathbf{P}(N = t)\\
&= \sum\limits_{t=0}^{\infty} t \cdot O(\log \log n_f) \cdot \mathbf{P}(N = t)\\
&\leq \sum\limits_{t=0}^{\infty} t \cdot O(\log \log n_f) \cdot \left(\frac{1}{2}\right)^t\\
&= O(\log \log n_f) \cdot \sum\limits_{t=0}^{\infty} t \cdot \left(\frac{1}{2}\right)^t\\
&= O(\log \log n_f) \cdot O(1)\\[2mm]
&= O(\log \log n_f)
\end{align*}
\end{proof}

%Theorem 4
\begin{theorem}
Algorithm \ref{alg:Bipartite2RulingSet} has an expected running time of
$O((\log \log n_f)^2 \cdot \log \log \min \{n_f, n_c\})$ rounds in the
$\mathcal{CONGEST}$ model.
\label{theorem:rulingset_runtime}
\end{theorem}
\begin{proof}
Once a certain level $L_i$ has been reached, the expected time for Algorithm
\textsc{Facility2RulingSet} to reach level $L_{i+1}$ is
$O((\log \log n_f) \cdot \log \log \min \{n_f, n_c\})$ (where the factor of
$O(\log \log \min \{n_f, n_c\})$ is the upper bound on the running time of the
dissemination subroutine). Since there are $O(\log \log n_f)$ levels to progress
through before reaching $L_{\ast}$ and terminating, the algorithm has an
expected running time of
$O((\log \log n_f)^2 \cdot \log \log \min \{n_f, n_c\})$ rounds.
\end{proof}

\section{Concluding Remarks}
Our expectation is that the Message Dissemination with Duplicates (MDD) problem
and its solution via probabilistic hashing will have applications in other distributed
algorithms in low-diameter settings.
This problem may also serve as a candidate for lower bounds research.
In particular, the results in this paper raise the question of whether $\Omega(\log\log \min\{n_c, n_f\})$
is a lower bound on the number of rounds it takes to solve MDD.
Alternately, it will be interesting (and surprising) to us if MDD was solved in $O(1)$ rounds.

Our two papers (the current paper and \cite{BHP12arxiv}) on super-fast algorithms yielding
$O(1)$-approximation for metric facility location, lead naturally to similar
questions for the non-metric version of the problem.
In particular, we are interested in super-fast algorithms, hopefully running in $O(\mbox{poly}(\log\log n))$ rounds,
that yield a logarithmic-approximation to the non-metric facility location problem on
cliques and complete bipartite networks.

\bibliographystyle{plain}
\bibliography{DistComp}

\end{document}